\documentclass[aps, pra, nofootinbib, 
twocolumn, 
notitlepage, superscriptaddress, longbibliography]{revtex4-1}
\newcommand{\papertitle}{Estimating the Coherence of Noise}
\newcommand{\pdfauthor}{Joel Wallman, Chris Granade, Robin Harper, Steven Flammia}

\usepackage{graphicx, amsmath, amssymb, amsthm, color, bm, bbm, pgfplots,dsfont}

\usepackage{times} 

\pgfplotsset{compat=newest}
\usepgfplotslibrary{external}
\theoremstyle{plain}
\newtheorem{thm}{Theorem}

\newtheorem{prop}[thm]{Proposition}

\theoremstyle{definition}

\theoremstyle{remark}

\renewenvironment{proof}[1][Proof]{\noindent\textbf{#1.} }{\ $\Box$}

\definecolor{nblue}{rgb}{0.2,0.2,0.7}
\definecolor{ngreen}{rgb}{0.1,0.5,0.1}
\definecolor{nred}{rgb}{0.8,0.2,0.2}
\definecolor{nblack}{rgb}{0,0,0}

\newcommand{\bs}[1]{\boldsymbol{#1}}

\newcommand{\G}{\mathcal{G}}
\newcommand{\E}{\mathcal{E}}
\newcommand{\M}{\mathcal{M}}

\newcommand{\unit}{\mathbbm{1}}

\newcommand{\mc}[1]{\mathcal{#1}}
\newcommand{\md}[1]{\mathds{#1}}
\newcommand{\bc}[1]{\boldsymbol{\mathcal{#1}}}
\newcommand{\mbb}[1]{\mathbb{#1}}

\newcommand{\inner}[1]{\langle #1\rangle}
\newcommand{\ket}[1]{|#1\rangle}
\newcommand{\bra}[1]{\langle #1|}

\newcommand{\tr}{\mathrm{Tr}}

\newcommand{\tn}[1]{^{\otimes #1}}

\newcommand{\av}[1]{\left\lvert #1\right\rvert}
\newcommand{\ct}{^{\dagger}}

\newcommand{\pa}[1]{\left(#1\right)}

\newcommand{\hidden}[1]{}

\newcommand{\mat}[2]{\left(\begin{array}{#1} #2 \end{array}\right)}

\usepackage{hyperref}
\definecolor{darkblue}{RGB}{0,0,127} 
\definecolor{darkgreen}{RGB}{0,150,0}
\hypersetup{breaklinks, colorlinks, linkcolor=darkblue, citecolor=darkgreen, filecolor=red, urlcolor=blue}
\hypersetup{pdftitle={\papertitle}, pdfauthor={\pdfauthor}}

\newcommand{\expect}{\md{E}}

\def\equationautorefname~#1\null{Eq.~(#1)\null}

\begin{document}
\title{\papertitle}

\author{Joel\ \surname{Wallman}}
\affiliation{Institute for Quantum Computing and Department of Applied 
Mathematics, University of Waterloo, Waterloo, Canada}
\author{Chris\ \surname{Granade}}
\author{Robin\ \surname{Harper}}
\author{Steven T.\ \surname{Flammia}}
\affiliation{Centre for Engineered Quantum Systems, School of Physics, The 
University of Sydney, Sydney, Australia}

\date{\today}

\begin{abstract}
	Noise mechanisms in quantum systems can be broadly characterized as either 
	coherent (i.e., unitary) or incoherent. For a given fixed average error 
	rate, coherent noise mechanisms will generally lead to a larger worst-case 
	error than incoherent noise. We show that the coherence of a noise source can be 
	quantified by the \emph{unitarity}, which we relate to the 
	average change in purity averaged over input pure states. We then show that the 
	unitarity can be efficiently estimated using a protocol based on 
	randomized benchmarking that is efficient and robust to state-preparation 
	and measurement errors. We also show that the unitarity provides a lower 
	bound on the optimal achievable gate infidelity under a given noisy process.
\end{abstract}

\maketitle


To harness the advantages of quantum information processing, quantum 
systems have to be controlled to within some maximum threshold error. 
Certifying whether the error is below the threshold is possible by performing 
full quantum process tomography~\cite{Chuang1997, Poyatos1997},  
however, quantum process tomography is both inefficient in the number of qubits 
and is sensitive to state-preparation and measurement errors 
(SPAM)~\cite{Merkel2013}.

Randomized benchmarking~\cite{Emerson2005, Knill2008, Dankert2009, Lopez2010, 
Magesan2011, Magesan2012a} and direct fidelity estimation~\cite{Flammia2011, 
da-Silva2011} have been developed as efficient methods for estimating the 
average infidelity of noise to the identity. However, the worst-case error, as 
quantified by the diamond distance from the identity, can be more relevant to 
determining whether an experimental implementation is at the threshold for 
fault-tolerant quantum computation~\cite{Kitaev1997}. The best possible bound 
on the worst-case error (without further information about the noise) scales as 
the square root of the infidelity and can be orders of magnitude greater than 
the reported average infidelity~\cite{Wallman2014,Sanders2015}. 

However, this scaling of the worst-case bound is only known to be saturated by 
unitary noise. 
If the noise is known to be stochastic Pauli noise, the worst-case error is 
directly proportional to the average infidelity~\cite{Magesan2012a}, vastly 
improving on the general bound. Consequently, quantifying the intermediate 
regime between unitary and fully incoherent noise may allow the bound on the
worst-case error to be substantially improved.

Randomized benchmarking is also emerging as a useful tool 
for diagnosing the noise in an experiment~\cite{OMalley2014, Wallman2014a}, 
which can then be used to optimize the implementation of gates by varying the 
experimental design. In this spirit, an experimental protocol for 
characterizing the coherence of a noise channel will be an important tool as 
the quest to build a fault-tolerant quantum computer progresses.

In this paper, we present a protocol for estimating a particular quantification 
of the coherence of noise, which we term the \emph{unitarity}, in the 
experimental implementation of a unitary 2-design. Our protocol is efficient 
and robust against SPAM, and is a minor modification of randomized 
benchmarking. The unitarity is defined as the average change in the 
purity of a pure state after applying the noise channel, with the contributions 
due to the identity component subtracted off, (see \autoref{def:unitarity}) and 
is closely related to the purity of the Jamio{\l}kowski isomorphic state (see 
Proposition \ref{prop:Jpure}). We show that the unitarity is invariant under unitary 
gates and attains its maximal value if and only if the noise is unitary. 
Furthermore, we show that the unitarity can be combined with the average gate 
fidelity to quantify how far a noise channel is from depolarizing noise. 
Finally, we show that the unitarity of a noise channel provides a lower bound 
on the best achievable gate infidelity assuming perfect unitary control.

Our approach to quantifying coherence complements other recent work on 
quantifying coherence since we focus on the coherence of quantum 
\textit{operations} rather than the coherence of quantum \textit{states} 
relative to a preferred basis~\cite{Baumgratz2014}.

\section{Defining Unitarity}\label{sec:unitarity}

We begin by defining the unitarity of a noise channel 
$\mc{E}:\mc{B}(\mbb{C}^d)\to\mc{B}(\mbb{C}^d)$, that is, a completely positive 
(CP) linear map that takes quantum states to quantum states. The purity of a 
quantum state $\rho$ is $\tr\rho\ct\rho\in[0,1]$ with $\tr\rho\ct\rho=1$ if and 
only if $\rho$ is a pure state. An initial candidate for a definition of the 
unitarity of $\mc{E}$ is
\begin{align}
\int {\rm d}\psi \tr[\mc{E}(\psi)\ct \mc{E}(\psi)],
\end{align}
that is, as the purity of the output states averaged over all pure state 
inputs. However, this definition is problematic, since it would lead to the 
nonunital state-preparation channel
\begin{align}\label{eq:E0}
\mc{E}_0(\rho) = \tr(\rho)\ket{0}\bra{0}
\end{align}
having the same value of unitarity as a unitary channel, even though it does 
not preserve coherent superpositions. Similarly, the (trace-decreasing) 
filtering channel
\begin{align}\label{eq:E1}
\mc{E}_1(\rho) = \ket{0}\bra{0}\rho\ket{0}\bra{0}
\end{align}
does not preserve coherent superpositions and so should have the same unitarity 
value as a complete depolarizing channel. Both of these problematic channels 
arise when either the identity is mapped to coherent terms or \emph{vice 
versa}. 

To avoid these issues, we define the unitarity of a noise channel to be the 
average purity of output states, with the identity components subtracted, 
averaged over all pure states. That is, we define
\begin{align}\label{def:unitarity}
u(\mc{E}) = \frac{d}{d-1}\int {\rm d}\psi \tr\bigl[\mc{E}'(\psi)\ct 
\mc{E}'(\psi)\bigr],
\end{align}
where the normalization factor is chosen so that $u(\mc{I})=1$ and $\mc{E}'$ is 
defined so that $\mc{E}'(A)=\mc{E}(A)-[\tr \mc{E}(A)/\sqrt{d}]\unit$ for all 
traceless 
$A$ (to account for trace-decreasing channels, such as in \autoref{eq:E0}) 
and $\mc{E}'(\unit_d)=0$ (to account for non-unital channels, such as in 
\autoref{eq:E1}). Equivalently, if $\{A_2,\ldots,A_{d^2}\}$ is any set of 
traceless and trace-orthonormal operators and with $A_1 = \unit/d$ 
(e.g., the normalized Paulis), then we can 
define the generalized Bloch vector ${\bf{n}}(\rho)$ of a density operator $\rho$ 
with unit trace to be the vector of $d^2-1$ expansion coefficients
\begin{align}
\rho = \unit_d/d + \sum_{k>1} n_k A_k\,.
\end{align}
Our definition of the unitarity is then equivalent to
\begin{align}
u(\mc{E}) = \frac{d}{d-1}\int {\rm d}\psi \bigl\lVert {\bf{n}}[\mc{E}(\psi)] 
- {\bf{n}}[\mc{E}(\unit_d/d)] \bigr\rVert^2,
\end{align}
that is, the average squared length (i.e., Euclidean norm) of the generalized 
Bloch vector after applying the map $\mc{E}$ with the component due to the 
identity subtracted off.

\section{The Estimation Protocol}\label{sec:protocol}

We now present a protocol for characterizing the unitarity of the 
noise in an experimental implementation of a unitary 2-design $\mc{G}$ under 
the assumption that the experimental implementation of any $U\in\mc{G}$ can be 
written as $\mc{U}\circ \mc{E}$ where $\mc{U}$ denotes the channel 
corresponding to conjugation by $U$ and $\mc{E}$ is a completely positive, 
trace-preserving (CPTP) channel independent of $U$. (Note that, as in all 
randomized benchmarking papers, the assumption that $\mc{E}$ is independent of 
$\mc{U}$ can be relaxed without dramatically effecting the 
results~\cite{Magesan2011, Wallman2014, Wallman2014a}.)

The protocol is to repeat many independent trials of the following.
\begin{itemize}
	\item Choose a sequence ${\bf{j}} = (j_1,\ldots,j_m)$ of $m$ integers in 
	$\mathbb{N}_{\av{\mc{G}}} = \{1,\ldots,\av{\mc{G}}\}$ uniformly at random.
	\item Estimate the expectation value $Q_{\bf{j}}$ of an operator $Q$ 
	after preparing the state $\rho$ and applying the sequence 
	$U_{\bf{j}}=U_{j_m}U_{j_{m-1}}\ldots U_{j_1}$ of operators. In the
	ideal case that $\mc{E} = \mc{I}$, the expectation value is given by 
	\begin{align}
		 Q_{\bf{j}} = \tr (Q U^{\vphantom{\dagger}}_{\bf{j}} \rho U_{\bf{j}}\ct).
	\end{align}
\end{itemize}

We will show in \autoref{sec:vegetables}, \autoref{thm:main}
that, under the above assumptions on the noise, 
the expected value of $Q^2_{\bf{j}}$ over all random sequences 
${\bf{j}}$ obeys
\begin{align}\label{eq:TP_decay}
	\md{E}_{\bf{j}}[Q_{\bf{j}}^2] = A + B u(\mc{E})^{m-1}
\end{align}
for trace-preserving noise, where $A$ and $B$ are constants incorporating SPAM 
and the nonunitality of the noise and $u(\mc{E})\in[0,1]$ is the unitarity of 
the noise defined in \autoref{def:unitarity}, with $u(\mc{E})=1$ if and only 
if $\mc{E}$ is unitary. 

Therefore estimating $\md{E}_{\bf{j}}[Q_{\bf{j}}^2]$ for multiple values of $m$
using the above protocol and fitting to \autoref{eq:TP_decay} gives an 
efficient and robust estimator of the unitarity.

\subsection{Estimators}

Note that, as opposed to standard presentations of randomized benchmarking, we
are considering the expectation of an operator $Q$ rather than the
probability of a single  outcome. Though these two descriptions are mathematically equivalent,
by presenting in terms of observables, we more easily take expectations over
multiple observables. For example, we can average over the non-identity Pauli operators,
while keeping the sequence the same. As will be discussed
in \autoref{sec:vegetables}, this allows us to simulate a two-state measurement
involving $S$, the SWAP operator.
We term this the \emph{purity measurement}, as it estimates the relevant 
state-dependent term in
\begin{equation}
	\label{eq:original_purity}
	\tr(\rho_{\mathbf{j}}^2) = \frac{1}{d}+\|\bf{n}(\rho_{\mathbf{j}})\|^2 \,,
\end{equation}
the purity of the state $\rho_{\mathbf{j}}$ produced by the sequence $\mathbf{j}$.
What we actually use is a shifted and rescaled version of this defined by
\begin{align}
	\label{eq:purity}
	P_{\mathbf{j}} = \frac{d}{d-1}\|\bf{n}(\rho_{\mathbf{j}})\|^2 \,,
\end{align}
which for physical states is always in the interval $[0,1]$. For a single qubit, 
and measuring in the Pauli basis, this quantity is just 
$P_{\mathbf{j}} = \langle X\rangle^2+\langle Y\rangle^2+\langle Z\rangle^2$,
where each expectation value is taken with respect to the state $\rho_{\mathbf{j}}$.

The purity measurement can be performed in one of two ways. The direct way 
involves using two copies of the experiment (with the same sequence) that are 
run in parallel and a SWAP gate applied immediately prior to measurement. 
A method using only one copy makes use of the expansion
\begin{align}
S = \sum_k A_k\otimes A_k\ct
\end{align}
for any orthonormal operator basis $\{A_k\}$ (e.g., the normalized Paulis) by adding up 
the expectation values over measurements in the operator basis for the same 
sequences, that is, by estimating $\sum_k \md{E}_{\mathbf{j}}[(A_k)_{\bf j}^2]$.

Implementing the purity measurement using this averaging reduces the between-sequence
contribution to the uncertainty in our estimates of $\expect_{\mathbf{j}}[Q_{\mathbf{j}}^2]$,
since if the noise is
approximately unitary, then the final state will be relatively pure but will
generically overlap with all non-identity Paulis.
We note, however, that the above summation over a trace-orthonormal basis is
not scalable with the number of qubits, since there are exponentially many
$n$-qubit Paulis. We leave possible optimizations and an analysis of the scalable 
two-copy protocol as an open problem.

Also note that unlike in standard randomized benchmarking, we do not require 
the unitary 2-design to be a group since we do not require an inverse operation, 
or even that the set $\mc{G}$ is closed under composition.

\subsection{Trace-decreasing noise}

More generally, some experimental noise $\mc{E}$ may be trace-decreasing 
with an average survival rate
\begin{align}
S(\mc{E}) = \int \mathrm{d}\psi \tr[\mathcal{E}(\psi)],
\end{align}
which is the amount of the trace of the quantum state $\psi$ that survives the 
error channel $\mc{E}$, averaged over the Haar measure $\mathrm{d}\psi$. 
When $\mc{E}$ is itself the average noise over $\mathcal{G}$, the average 
loss rate can be estimated by
\begin{align}\label{eq:leakage}
\md{E}_{\bf{j}}[Q_{\bf{j}}] = C S(\mc{E})^{m-1}
\end{align}
where $C$ is a constant determined by SPAM~\cite{Wallman2014a}.

For trace-decreasing noise, the standard decay curve in 
\autoref{eq:TP_decay} can be generalized to
\begin{align}\label{eq:TD_decay}
\md{E}_{\bf{j}}[Q_{\bf{j}}^2] &=  A\lambda_+^{m-1} + B \lambda_-^{m-1},
\end{align}
for some constants $A$ and $B$ where
\begin{align}\label{eq:eigenvalue_sum}
\lambda_+ + \lambda_- = S(\mc{E})^2 + u(\mc{E}).
\end{align}

The above protocol is a variation of standard randomized benchmarking 
experiments, and is very similar to the protocol for estimating loss presented 
in Ref.~\cite{Wallman2014a}. In particular, one estimates an exponential decay 
rate in an exactly analogous manner (see \autoref{eq:TP_decay}) and the 
result is obtained in a manner that is robust to SPAM. 

However, there are three small but crucial differences to the experimental 
protocol presented in Ref.~\cite{Wallman2014a}, leading to significant 
differences in the analysis and interpretation of the decay curves. Most 
importantly, the post-processing is different, since in the present 
paper the survival probabilities for the individual sequences are squared 
before they are averaged. Secondly, the preparation and measurement 
procedures in the loss protocol of Ref.~\cite{Wallman2014a} are 
ideally the maximally mixed state and the trivial (identity) measurement respectively,
which is out performed in this protocol by the use of the purity measurement.
Finally, the 
current protocol requires a unitary 2-design, whereas the loss 
protocol only requires a unitary 1-design (although it also works for a unitary 
2-design).

\section{Numerical Simulations}\label{sec:numerics}

We now illustrate our model and our experimental protocol by numerically simulating each for
a variety of single-qubit noise models.
In \autoref{fig:constant_unitary}, we give an example of the correctness of our model \autoref{eq:TP_decay}
by showing that it agrees with simulated experimental data in the extreme case that the 
error channel is a fixed unitary. In \autoref{fig:nonunital_noise}, we demonstrate the utility 
of \autoref{eq:TP_decay} as an estimation model by estimating $u(\mathcal{E})$ from 
simulated data drawn according to our protocol. We simulate measurement error on each 
measurement with small independent random orthogonal matrices, scaling the unital 
components with a random factor between 0.95 and 1.0. In both these simulations we 
simulate SPAM on the prepared state by applying a random near-identity unitary. We 
choose $\mc{G}$ to be the single-qubit Clifford group.

\begin{figure}
\centering
\newlength\figureheight
\newlength\figurewidth
\setlength\figureheight{6cm}
\setlength\figurewidth{7cm}
\includegraphics[width=\figurewidth,height=\figureheight]{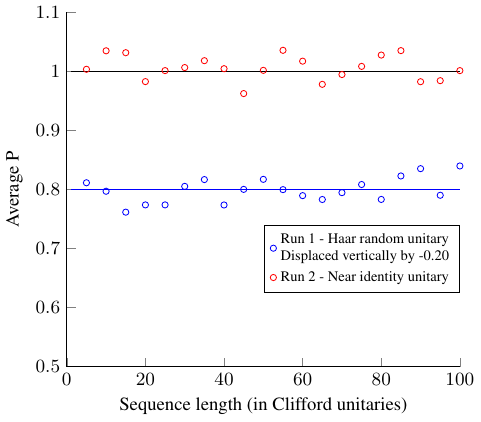}
	\caption{(Color online)
	Comparison of our model (\autoref{eq:TP_decay}) with simulated data,
	where the error channel is taken to be a fixed unitary drawn
	from the Haar measure (blue), or given by a rotation of $0.1$ about $\sigma_x$ (red).
	In the simulated data, the purity is measured following the formula of \autoref{eq:purity}.
	The solid lines show the model \autoref{eq:TP_decay} for the calculated values of
	$A$, $B$ and $u(\mc{E})$ for each scenario. This shows that even for extreme 
	values of unitary noise, our model correctly predicts the unitarity. 
	}\label{fig:constant_unitary}
\end{figure}

Concretely, in \autoref{fig:constant_unitary} we show two runs. In the first we set 
$\mc{E}$ to be some fixed (systematic) unitary chosen randomly according to the 
Haar measure (\emph{Haar-random unitary}) and some near-identity unitary 
represented by a rotation of $0.1$ radians around the $X$-axis of the Bloch sphere 
(\emph{near-identity unitary}). As \autoref{fig:constant_unitary} demonstrates, our model is 
in indeed insensitive to unitary noise. 

In \autoref{fig:nonunital_noise}, we show different types of unital noise 
composed with the nonunital amplitude-damping channel
\begin{align}\label{eq:nonunital_channel}
\mc{E}_{\rm{d}}(\rho) = p\ket{0}\bra{0} + (1-p)\rho
\end{align}
to simulate relaxation to a ground state. The particular unital channels we 
consider are a Haar-random unitary and a gate-dependent 
noise channel corresponding to choosing a fixed perturbation of the eigenvalues 
of a unitary $g$ by $e^{i\epsilon}$ to simulate over/under-rotation errors, where the 
perturbations $\epsilon$ are chosen independently and uniformly from $[-0.1,0.1]$ 
radians for each gate (\emph{rotation channel}). 

\begin{figure}
\centering
\setlength\figureheight{6cm}
\setlength\figurewidth{7cm}
\includegraphics[width=\figurewidth,height=\figureheight]{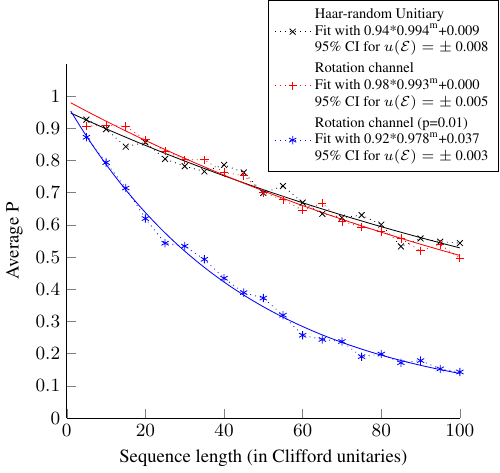}
	\caption{(Color online)
		Numerical purity decay curve for our protocol 
		when the noise consists of the non-unital channel in
		\autoref{eq:nonunital_channel} with $p=0.003$ composed with: a Haar-random unitary (black 
		crosses), and fixed gate-dependent over/under-rotations where the eigenvalues 
		for each gate are perturbed by $e^{\pm i\delta}$ for some 
		$\delta\in[-0.01,0.01]$ chosen independently and uniformly (red pluses).
		The purity is measured in the same manner as in 
		\autoref{fig:constant_unitary}.
		The third 
		plot (blue stars) shows (gate-dependent) over/under-rotations 
		composed with the 
		nonunital channel with p=$0.01$. The solid lines give the fit to 
		\autoref{eq:TP_decay}, where the slope gives an estimate for $u(\mc{E})$ 
		given by the values $0.994$, $0.993$ and $0.978$, 
		consistent with the theoretical values of 
		$0.994$, $0.994$ and $0.981$, respectively.
		Confidence intervals for $u(\mc{E})$ were calculated assuming Gaussian noise, such that the least-squares fit is a maximum likelihood estimator.
	}
	\label{fig:nonunital_noise}
\end{figure}

Note that the statistical fluctuations in Figs.~\ref{fig:constant_unitary} and \ref{fig:nonunital_noise} arise from 
between-sequence variations and within-sequence variations. The 
between-sequence variations arise from sampling a 
small number of random sequences (30 sequences in this case) relative to the total number. The between-sequence variations are minimised by measuring an observable for the purity. A perturbation 
expansion of the form $\bc{E} = \bc{I} - r\bs{\delta}$ (where $r$ is the 
average gate 
infidelity of $\bc{E}$ to the identity) together with appropriate bounds on the 
diamond norm can be used to bound these fluctuations and show that they must 
decrease with gate infidelity, as in Ref.~\cite{Wallman2014}. However, 
a more 
detailed analysis is complicated by the complexity of the relevant 
representation theory (that is, four-fold tensor products). The within-sequence 
variations arise from the need to estimate the expectation values of the 
observables. For the purpose of Figs.~\ref{fig:constant_unitary} and \ref{fig:nonunital_noise}, we used an unbiased estimator of 
the squared expectation values, simulating $N$ measurements (with $N$ set to 150).

\begin{figure}
\centering
\includegraphics[width=\columnwidth]{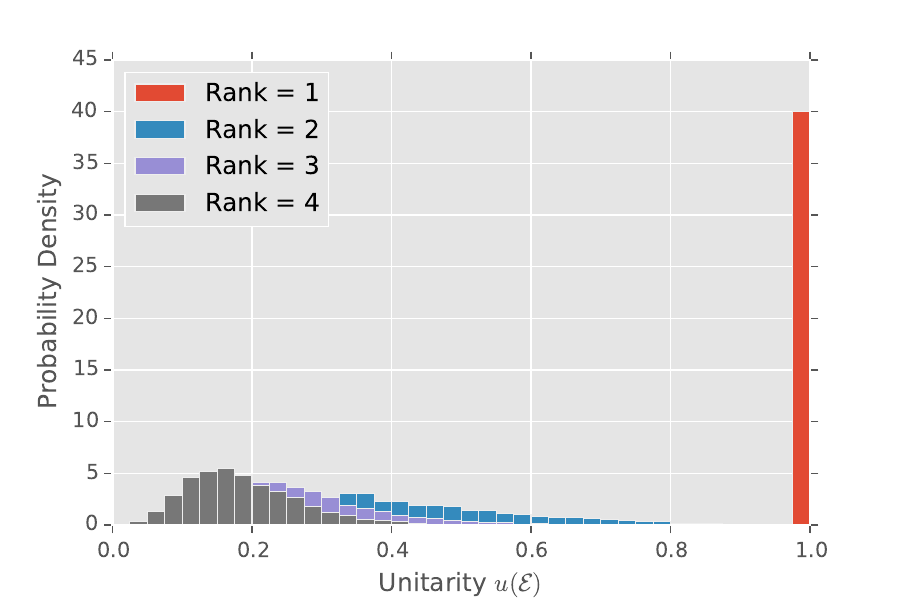}
	\caption{(Color online)
		Unitarity of single-qubit CPTP channels chosen according to the random 
		distributions of Bruzda \emph{et al}.~\cite{bruzda_random_2009} with
		varying ranks of the Kraus operators, demonstrating that the unitarity
		carries information about the structure of the channel. In particular,
		channels which require more Kraus operators to specify tend towards
		much smaller unitarity.
	}\label{fig:unitarity_v_rank}
\end{figure}

Finally, we consider the unitarity of random channels drawn from the random 
ensemble
of Bruzda \emph{et al}.~\cite{bruzda_random_2009}, using the QuTiP software 
package~\cite{johansson_qutip_2013}
to draw channels and compute their unitarity (see Supplemental Material).
As shown in \autoref{fig:unitarity_v_rank},
the distribution of unitarities depends strongly on the Kraus rank of the random channel. 
Moreover, as demonstrated in \autoref{fig:unitarity_v_agf}, this information is 
correlated with, but distinct from, the average gate fidelity.

\begin{figure}
\centering
\includegraphics[width=\columnwidth]{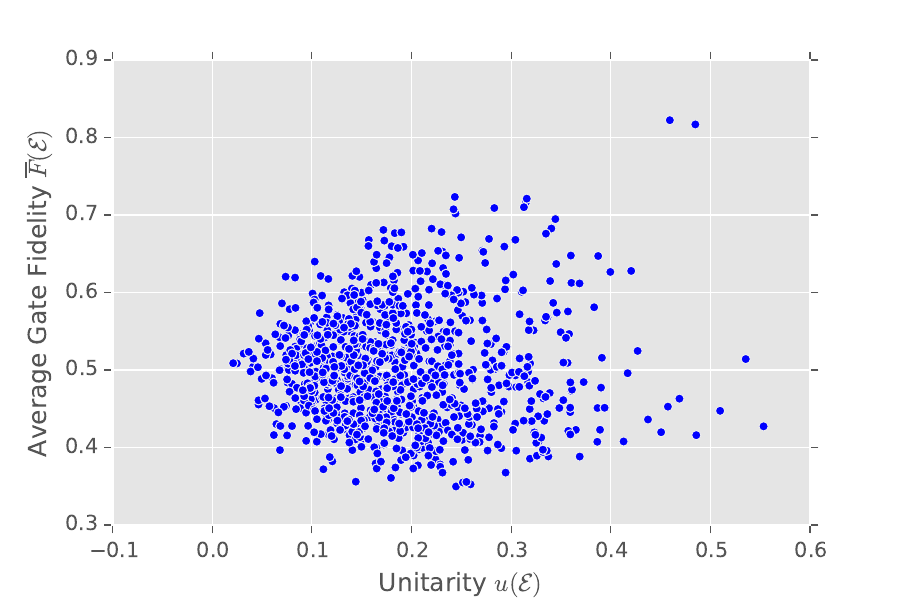}
	\caption{(Color online)
		Unitarity of single-qubit CPTP channels chosen according to the random 
		distributions of Bruzda \emph{et al}.~\cite{bruzda_random_2009},
		plotted versus their fidelities to the identity 
		channel. This example shows that even though the two quantities are
		correlated, they are not redundant and give different insight into the
		structure of the noise.
	}\label{fig:unitarity_v_agf}
\end{figure}

\section{Derivation of the Fit Models}\label{sec:vegetables}

We now derive the decay curve in \autoref{eq:TD_decay} for 
trace-non-increasing noise and show how the decay curve in \autoref{eq:TP_decay} 
emerges as a special case for trace-preserving noise.

Since we are dealing with sequences of channels, it will be convenient to work 
in the Liouville representation. Since a quantum channel is a linear map 
between finite-dimensional 
vector spaces, it is always possible to represent it as a matrix acting on 
basis coefficients in some given bases for the vector spaces.

In order to construct the Liouville representation of channels, let $\mc{A} = 
\{A_1,\ldots,A_{d^2}\}$ be an orthonormal basis of $\mbb{C}^{d\times d}$ 
according to the Hilbert-Schmidt inner product $\inner{A,B}=\tr A\ct B$. Any 
density matrix $\rho$ can be expanded as $\rho = \sum_{k\in\mbb{N}_{d^2}} 
\inner{A_k,\rho}A_k$ and so we can identify $\rho$ with a column vector 
$|\rho)\in\mbb{C}^{d^2}$ whose $k$th entry is $\inner{A_k,\rho}$. The Liouville 
representation of a channel $\mc{E}$ is then the unique matrix 
$\bs{\E}\in\mbb{C}^{d^2\times d^2}$ such that $\bs{\E}|\rho) = |\E[\rho])$, 
which has entries $\bc{E}_{kl}=\inner{A_k,\mc{E}(A_l)} = (A_k|\bc{E}|A_l)$. An 
immediate consequence of the uniqueness of $\bc{E}$ is that the composition of 	
abstract maps is represented in the Liouville representation by matrix 
multiplication.

The Liouville representation of unitary channels forms a unitary projective 
representation of the unitary group $U(d)$. When we wish to emphasize the 
Liouville representation as a formal representation (rep) of the unitary group 
$U(d)$ (or subgroups thereof), we will use the notation $\phi_L(U)$ instead of 
$\bc{U}$. With this notation, it is easy to verify that $\phi_L$ is indeed a 
unitary representation of $U(d)$, since the Liouville representation of 
composition is matrix multiplication and it can easily be verified that 
$\phi_L(U\ct) =\phi_L(U)\ct$.

Any representation $\phi$ of a semisimple group $\mc{G}$ [such as $SU(d)$] over 
a vector space $V$ can be unitarily decomposed into a direct sum of irreducible 
representations (irreps) $\bigoplus_l \phi_l\otimes \unit_{n_l}$, where the $l$ 
label the irreps and the $n_l$ are the corresponding multiplicities and a rep 
$\phi$ over a vector space $V$ is irreducible if there are no nontrivial 
subspaces of $V$ that are invariant under the action of $\phi$. A particularly 
important irrep for this paper is the trivial irrep $\phi_T$ such that 
$\phi_T(g)=1$ for all $g\in\mc{G}$. 

In the Liouville representation, vectors $b\in\mbb{C}^{d^2}$ are in one-to-one 
correspondence with operators $B\in\mbb{C}^{d\times d}$, so invariant (vector) 
subspaces under the Liouville representation can be identified with operator 
subspaces that are invariant under conjugation in the canonical (i.e., $d\times 
d$ matrix) representation. In particular, the identity operator $\unit$ is 
invariant under conjugation by any unitary, so $|\unit)$ is an invariant 
subspace of the Liouville representation corresponding to a trivial irrep. 
We now fix $A_1 = \unit/\sqrt{d}$ (so that $\inner{A_1,A_1}=1$), so that 
the Liouville representation of any unitary $\mc{U}$ is 
\begin{align}\label{eq:single_rep}
\phi_L(U) = 1\oplus \phi_{\rm u}(U) 
\end{align}
where $\oplus$ denotes the matrix direct sum and we refer to $\phi_{\rm u}(U)$ 
as the \textit{unital irrep}, which has dimension $d^2-1$. Furthermore, any CP 
channel $\E$ can be written in a corresponding block form as
\begin{align}\label{eq:block_terms}
\bs{\E} = \mat{cc}{S(\mc{E})  & \bs{\E}_{\rm sdl} \\ \bs{\E}_{\rm n} 
& \bs{\E}_{\rm u}},
\end{align}
where we refer to $\bc{E}_{\rm sdl}$, $\bs{\E}_{\rm n}$ and 
$\bs{\E}_{\rm u}$ as the \textit{state-dependent leakage}, \textit{nonunital} 
and \textit{unital} blocks respectively. We now show how $\bc{E}_{\rm u}$ is 
related to the definition of the unitarity in \autoref{def:unitarity}.

\begin{prop}\label{prop:unitarity_relation}
The unitarity of a channel $\mc{E}$ is
\begin{align}
u(\mc{E}) = \frac{1}{d^2-1}\tr \bc{E}_{\rm u}\ct\bc{E}^{\vphantom{\ct}}_{\rm u}
\end{align}
\end{prop}

\begin{proof}
For any operator $A$, $\tr A\ct A = (A|A)$ and $\bc{E}'=P_{\rm u}\bc{E}P_{\rm 
u}$ where $P_{\rm u}$ is the projector onto the unital irrep, so 
\autoref{def:unitarity} can be rewritten as
\begin{align}
u(\mc{E}) &= \frac{d}{d-1}\int {\rm d}\psi 
(\psi|P_{\rm u}\bc{E}\ct P_{\rm u}\bc{E}P_{\rm u}|\psi)	\notag\\
& =  \frac{d}{d-1}\tr \bc{E}_{\rm u}\ct \bc{E}^{\vphantom{\ct}}_{\rm u} \mc{O}\,,
\end{align}
where $\mc{O} = \int {\rm d}\psi|\psi)(\psi|$, with the slight abuse of 
notation $\bc{E}_{\rm u} = P_{\rm u} \bc{E} P_{\rm u}$. 

Since $\mc{O}$ commutes with the action of the unitary group, Schur's lemma 
implies that it is a weighted sum of projectors onto the irreps of $\phi_L$,
\begin{align}
	\mc{O} = \lambda_T P_T + \lambda_{\rm u} P_{\rm u} \,.
\end{align}
The projector onto the trivial irrep is $P_T = |A_1)(A_1|$ and so 
\begin{align}
 \frac{\tr P_T \mc{O}}{\tr P_T} = \lambda_T = \int {\rm d}\psi \lvert(A_1|\psi)\rvert^2 = \frac1d.
\end{align}
Because $\tr \mc{O} = 1$ from the normalization of the Haar measure, 
we can solve for $\lambda_{\rm u}$ in the expression
\begin{align}
\tr \mc{O}  =  1 = \frac{1}{d} (1) + \lambda_{\rm u} (d^2-1) \,
\end{align} 
and we find $\lambda_{\rm u} = 1/d (d+1)$. Plugging this in and using $P_{\rm 
u}P_{T} = 0$ gives the final result.
\end{proof}

Before we derive the decay curve in \autoref{eq:TD_decay} using the 
expression 
for the unitarity from Proposition~\ref{prop:unitarity_relation}, let us first 
simplify the quantity of interest.  The expectation value of $Q$ given that 
the sequence $\bf{j}$ was applied is
\begin{align}
Q_{\bf{j}} = (Q|\bc{U}_{j_m}\bc{E}\ldots \bc{U}_{j_2}\bc{E}\bc{U}_{j_1}|\rho),
\end{align}
where a residual noise term has been absorbed into the experimental state 
preparation $\rho$. Noting that
\begin{align}
	Q_{\bf{j}}^2 = (Q\tn{2}|\bc{U}_{j_m}\tn{2}\bc{E}\tn{2}\ldots 
	\bs{\mc{U}}_{j_2}\tn{2}\bc{E}\tn{2}\bc{U}_{j_1}\tn{2}|\rho\tn{2}),
\end{align}
the expected average of the squares is
\begin{align}\label{eq:mean_squares}
\md{E}_{\bf{j}}[Q_{\bf{j}}^2] &= \av{\mc{G}}^{-m}\sum_{\bf{j}} Q_{\bf{j}}^2 \notag\\
&=(Q\tn{2}|\pa{\bc{U}\tn{2}_{\rm avg} \bc{E}\tn{2}\bc{U}\tn{2}_{\rm avg} }^{m-1}
|\rho\tn{2})\notag\\
&=(Q\tn{2}|\bs{\M}^{m-1}|\rho\tn{2}),
\end{align}
where $\bs{\mc{U}}\tn{2}_{\rm avg} = 
\av{\mc{G}}^{-1}\sum_{g\in\mc{G}}\bs{g}\tn{2}$, we define the averaged operator 
$\bs{\M} = \bc{U}\tn{2}_{\rm avg} \bc{E}\tn{2}\bc{U}\tn{2}_{\rm avg}$, and we 
have used the fact that
$\av{\G}^{-1} \sum_{g\in\G} \phi(g)$ is the projector onto the 
trivial subreps for any rep $\phi$ of a group $\mc{G}$ so that 
$\bc{U}\tn{2}_{\rm avg}=(\bc{U}\tn{2}_{\rm avg})^2$~\cite{Goodman2009}.
Thus, to derive the fit model we must first identify 
the trivial irreps of $\mc{G}$ in $\phi_L(U)\tn{2}$, since this is where $\bs{\M}$ is supported. 

\begin{prop}
The averaged operator $\bs{\M} = \bc{U}\tn{2}_{\rm avg} \bc{E}\tn{2}\bc{U}\tn{2}_{\rm avg}$ is supported 
on a two-dimensional subspace spanned by $|\unit_{d^2})$ and $|S)$, where $S$ is the 
SWAP operator.
\end{prop}

\begin{proof}
Define $\chi_R(g) = \tr R(g)$ as the character of the rep $R$. 
Then we can use Schur's orthogonality relations to count the number of trivial 
irreps.
Let $\langle\!\langle \chi_R , \chi_{R'}\rangle\!\rangle =\av{\mc{G}}^{-1}\sum_{g\in\mc{G}} \chi^*_R(g) \chi_{R'}(g)$ 
denote the character inner product for $\mc{G}$. From the direct sum structure 
in \autoref{eq:single_rep}, the number of trivial irreps is 
\begin{align*}
\langle\!\langle \chi_1, \chi_{L\otimes L} \rangle\!\rangle = 
\langle\!\langle \chi_1, \chi_1+2 \chi_{\rm u} + \chi^2_{\rm u}\rangle\!\rangle =  
1 + 0 + \langle\!\langle \chi_1 , \chi^2_{\rm u}\rangle\!\rangle \,.
\end{align*}
Since $\chi_{\rm u}$ is real-valued, we have 
$\langle\!\langle \chi_1 , \chi^2_{\rm u}\rangle\!\rangle = 
\langle\!\langle \chi_{\rm u} , \chi_{\rm u}\rangle\!\rangle$. 
If $\mc{G}$ acts irreducibly on the unital block~\cite{Gross2007a}, then 
$\langle\!\langle \chi_{\rm u} , \chi_{\rm u}\rangle\!\rangle = 1$
and the number of trivial irreps is 2. 

The two trivial irreps in $\phi_L\tn{2}$ are spanned by the orthonormal vectors 
$|B_1)$ and $|B_2)$ where
\begin{align}\label{eq:Bbasis}
	B_1 = \unit_{d^2}/d \quad , \quad B_2 = (S-\unit_{d^2}/d)/\sqrt{d^2-1}\,,
\end{align}
and $S$ is the SWAP operator. To check this, note that
\begin{align*}
(B_j|\phi_L(U)\tn{2}|B_k) &= \tr\bigl(B_j\ct U\otimes U B_k U\ct\otimes U\ct\bigr) = \delta_{jk}\,,
\end{align*}
since both identity and SWAP are invariant under conjugation by $U\otimes U$ and 
$S^2 = \unit_{d^2}$. 
Since $\phi_L(U)\tn{2}$ is a unitary rep, $B_1$ 
and $B_2$ are the first two elements of a two-qudit orthonormal Schur basis
$\{B_j\}$ for $\phi_L\tn{2}$ and so correspond to trivial irreps.
Therefore $\bs{\M}$
is zero except for the $2\times 2$ submatrix supported on $|B_1)$ and $|B_2)$. 
These vectors have the same span as $|\unit_{d^2})$ and $|S)$.
\end{proof}

The next proposition characterizes the averaged operator on the supported subspace.
\begin{prop}\label{prop:norms}
In the invariant $B_i$ basis from \autoref{eq:Bbasis}, the averaged operator $\bs{\M}$
has the following matrix elements
\begin{itemize}
\item $\bs{\M}_{11} = S(\mc{E})^2$, 
\item $\bs{\M}_{12} = (d^2-1)^{-1/2}\|\bs{\E}_{\rm sdl}\|^2$, 
\item $\bs{\M}_{21} = (d^2-1)^{-1/2}\|\bs{\E}_{\rm n}\|^2$, and 
\item $\bs{\M}_{22} = (d^2-1)^{-1} \|\bs{\E}_{\rm u}\|_F^2 = u(\E)$\,.
\end{itemize}
\end{prop}

\begin{proof}
We will establish the matrix elements with respect to $|\unit_{d^2})$ and $|S)$, and the 
claims about the $B_i$ basis will follow by taking appropriate linear combinations. 

Because the $B_i$ basis is invariant, we can ignore the average unitary terms 
in $\bs{\M}$. 
We first find that
\begin{align}
	(\unit_{d^2}|\bc{E}\tn{2}|\unit_{d^2}) & = \tr \E\tn{2}(\unit_{d^2}) = 
	\tr \bigl(\E(\unit_{d})\tn{2}\bigr) \notag\\
	& = \bigl(\tr \E(\unit_{d})\bigr)^2 = d^2 S(\mc{E})^2\,.
\end{align}
Next we can use the identity 
$\langle S,A\otimes B\rangle = \tr\bigl[S (A\otimes B)\bigr] = \tr(AB)$ 
and the fact that $\E(A^\dagger) = \E(A)^\dagger$ to find
\begin{align}
	(S|\bc{E}\tn{2}|\unit_{d^2}) &= \tr\bigl[S \E(\unit_d)\tn{2}\bigr] = \tr\bigl[\E(\unit_d) \E(\unit_d)\bigr] \notag\\
	& = \tr\bigl[\E(\unit_d)^\dagger \E(\unit_d)\bigr] = \|\E(\unit_d)\|_F^2\,,
\end{align}
where $\|\cdot\|_F$ denotes the Frobenius norm. The expression for 
$(\unit_{d^2}|\bc{E}\tn{2}|S)$ follows similarly using the adjoint channel. 
Finally, we can use the expansion $S = \sum_k A_k \otimes A_k^\dagger$ for any 
orthonormal operator basis $A_k$ to obtain
\begin{align}\label{eq:S-EE-S}
	(S|\bc{E}\tn{2}|S) & = \tr\bigl[S \bc{E}\tn{2}(S)\bigr] 
	= \sum_k \tr\bigl[S \E(A_k)\otimes\E(A_k^\dagger)\bigr]\notag \\
	& = \sum_k \tr\bigl[\E(A_k)^\dagger \E(A_k)\bigr] = \|\bc{\E}\|_F^2\,.
\end{align}
The values of the matrix elements are then established by using the form of 
\autoref{eq:block_terms} and the definition of the $B_i$ basis from 
\autoref{eq:Bbasis} and taking various linear combinations.
\end{proof}

The final step in deriving the fit model is to analyze the eigenvalues and eigenvectors 
of the averaged operator.
\begin{prop}\label{prop:eigenvec}
The averaged operator $\bs{\M}$ has
two distinct nontrivial eigenvectors. 
\end{prop}
\begin{proof}
Since the averaged operator vanishes almost everywhere, we only need to 
consider the 
$2 \times 2$ submatrix derived above. The nontrivial eigenvalues are
\begin{align}\label{eq:eigenvalues}
\lambda_{\pm} &= \frac{1}{2}\bs{\M}_{11} + \frac{1}{2}\bs{\M}_{22}\notag\\
&\quad \pm\frac{1}{2} \sqrt{[\bs{\M}_{11} - \bs{\M}_{22}]^2 + 
4\bs{\M}_{12} \bs{\M}_{21}}\,.
\end{align}
This spectrum is degenerate precisely when the terms under the square root 
both vanish (since both terms are nonnegative). Whenever the spectrum is nondegenerate, 
there are trivially two distinct eigenvectors, so we only need to deal with the 
degenerate case. 

We will break the analysis for the degenerate spectrum into two nontrivial cases, 
$\bs{\M}_{11}=\bs{\M}_{22}$ and either $\bs{\M}_{12}=0$ or $\bs{\M}_{21}=0$, exclusive. 
There are also two trivial cases: when $\bs{\M}_{12}=\bs{\M}_{21}=0$, the matrix $\bs{\M}$ 
is already diagonal and we are done. We ignore the pathological case when 
$\bs{\M}_{11} = 0$, since this corresponds physically to a state that is never observable. 
In both nontrivial cases, we will make use of the two-qudit state 
$\Pi_{a} = \frac{\unit-S}{d (d-1)}$, the maximally mixed state on the antisymmetric subspace. 
Expanding this state in the $B_i$ basis gives
\[|\Pi_a) = \pi_1 |B_1) + \pi_2 |B_2) = \frac{1}{d} |B_1) - \frac{\sqrt{d^2-1}}{d (d-1)} |B_2) \,. \]
The key feature of this state is that $\pi_2 < 0$. 

Case 1: $\bs{\M}_{12}=0$. In this case,
\begin{align}
\bs{\M} = \left(\begin{array}{cc} \lambda & 0 \\ y & \lambda\end{array}\right),
\end{align}
where $\lambda > 0$ and $y\ge0$. Taking the $m$th power gives
\begin{align}
\bs{\M}^m = \lambda^{m-1}\left(\begin{array}{cc} \lambda & 0 \\ m y & \lambda\end{array}\right)\,.
\end{align}
If we perform the measurement $\{\Pi_a,\unit-\Pi_a\}$ on a system prepared in 
the state 
$\unit_{d^2}/d^2$ which evolves under $\bs{\M}^m$, then the probability of observing the 
outcome $\Pi_a$ is
\begin{align}
(\Pi_a|\bs{\M}^m|\unit_{d^2}/d^2) &= 
\frac{\lambda^{m-1}}{d} (\lambda \pi_1 + m y \pi_2).
\end{align}
Since $\lambda, \pi_1>0$, $y \ge 0$, and $\pi_2 < 0$, in order for this to be a probability 
for all $m$, we require $y=0$ and so $\bs{\M}$ is actually diagonal.

Case 2: $\bs{\M}_{21}=0$. In this case,
\begin{align}
\bs{\M} = \left(\begin{array}{cc} \lambda & y \\ 0 & \lambda\end{array}\right),
\end{align}
where $\lambda > 0$ and $y\ge0$. Taking the $m$th power gives
\begin{align}
\bs{\M}^m = \lambda^{m-1}\left(\begin{array}{cc} \lambda & my \\ 0 & \lambda\end{array}\right).
\end{align}
Therefore the probability of detecting the system (i.e., measuring $\unit_{d^2}$) 
when a system is prepared in the state $\Pi_a$ and evolves under $\bs{\M}^m$ is
\begin{align}
(\unit_{d^2}|\bs{\M}^m|\Pi_a) = \lambda^{m-1}(\lambda \pi_1+m y \pi_2)\,.
\end{align}
Again since $\lambda, \pi_1>0$, $y \ge 0$, and $\pi_2 < 0$, for this to be a valid 
probability for all $m$, we require $y=0$ and so $\bs{\M}$ is actually diagonal.
\end{proof}

We now have all the ingredients to derive the fit models of 
Eqs.~\eqref{eq:TP_decay} and \eqref{eq:TD_decay}. 

\begin{thm}\label{thm:main}
For time- and gate-independent noise, the expected value $\md{E}_{\bf{j}}[Q_{\bf{j}}^2]$
obeys the decay equation
\[\md{E}_{\bf{j}}[Q_{\bf{j}}^2] = A + B u(\mc{E})^{m-1}\]
for trace-preserving noise, and for trace-decreasing noise it obeys
\[\md{E}_{\bf{j}}[Q_{\bf{j}}^2] =  A\lambda_+^{m-1} + B \lambda_-^{m-1}\,,\]
where $\lambda_\pm$ are given by \autoref{eq:eigenvalues}, 
$\lambda_+ + \lambda_- = S(\mc{E})^2 + u(\mc{E})$,
and the constants
$A$ and $B$ depend only on state preparation and measurement errors and the 
unitary that diagonalizes $\bc{M}$.
\end{thm}

\begin{proof}
Proposition \ref{prop:eigenvec} establishes that the matrix $\bs{\M}$ is diagonalizable 
by a similarity transform with eigenvalues given by \autoref{eq:eigenvalues}.
From \autoref{eq:mean_squares}, we can diagonalize $\bs{\M}$ and absorb 
the similarity transform into $|\rho\tn{2})$ and $(Q\tn{2}|$ as SPAM, yielding
\begin{align}
	\md{E}_{\bf{j}}[Q_{\bf{j}}^2] &=(Q\tn{2}|\bs{\M}^{m-1}|\rho\tn{2})\notag \\
	&= A\lambda_+^{m-1} + B\lambda_-^{m-1}\,.
\end{align}
Trace-preserving noise is a special case of this, since 
if $\mc{E}$ is TP, then by Proposition~\ref{prop:norms} we have
$\lambda_+ = S(\mc{E})^2 = 1$ and so $\lambda_- = u(\mc{E})$. 
\end{proof}

We note that the unitary that diagonalizes $\bs{\M}$ will in general depend on the 
noise channel, and hence will depend on $u$. We conflate this dependence with the 
SPAM errors in our fit model, as the diagonalization of $\bs{\M}$ does not depend
on the sequence length $m$. Neglecting the dependence on $u$ thus results in a model
that is correct, but is slightly less sensitive to $u$ than is optimal.

We are now equipped to formalize the observation made in \autoref{sec:protocol}
that the optimal observable for trace-preserving noise is an operator proportional
to $B_2$. This follows from noting that such operators overlap fully with the component of 
$\bc{M}$ that give rise to the exponential term, as given by \autoref{eq:S-EE-S}.

\section{Properties of the Unitarity}\label{sec:properties}

We now prove some properties of the unitarity for CPTP channels that 
make it a practical quantification of the coherence of a channel. We begin by 
proving that the unitarity and the average incoherent survival probability can 
be used to bound the nonunital and state-dependent leakage terms which are 
subtracted off in the definition of unitarity in \autoref{def:unitarity}. 

\begin{prop}\label{prop:bound_norms}
For any channel $\mc{E}$, 
\begin{align}\label{eq:bound_norms}
\max\{\lVert \bc{E}_{\rm n}\rVert^2 , \lVert \bc{E}_{\rm sdl}\rVert^2\} \leq 
\tfrac12(d^2-1)[S(\mc{E})^2-u(\mc{E})].
\end{align}
If $\mc{E}$ is trace-preserving, then $\lVert \bc{E}_{\rm n}\rVert^2 \leq 
(d-1)[1-u(\mc{E})]$.
\end{prop}

\begin{proof}
Consider the maximally mixed states on the symmetric and antisymmetric 
subspaces, $\Pi_s = \tfrac{\unit+S}{d(d+1)}$ and $\Pi_a = 
\tfrac{\unit-S}{d(d-1)}$ respectively, and let $E_s$ and $E_a$ be the respective 
projectors onto these spaces. Expanding these states in the $B_i$ 
basis gives
\begin{equation}
\begin{aligned}
|\Pi_s) &= \frac{1}{d}|B_1) + \frac{\sqrt{d^2-1}}{d(d+1)}|B_2) \\
|\Pi_a) &= \frac{1}{d}|B_1) - \frac{\sqrt{d^2-1}}{d(d-1)}|B_2)\,.
\end{aligned}
\end{equation}
Preparing the state $\Pi_s$ ($\Pi_a$), evolving under $\bc{M}$ and then 
measuring the POVM $\{E_a,\unit-E_a\}$ ($\{E_s,\unit-E_s\}$) produces 
the outcomes $E_a$ ($E_s$) with probabilities
\begin{align}
p_{as} &= (E_a|\bc{M}|\Pi_s) \notag\\
&=\frac{d-1}{2 d}\left(S(\mc{E})^2
-u(\mc{E})- \frac{\lVert \bc{E}_{\rm n} \rVert^2}{d-1} + \frac{\lVert 
\bc{E}_{\rm sdl} \rVert^2}{d+1}\right)\label{eq:anti_sym}\\
p_{sa}&=(E_s|\bc{M}|\Pi_a) \notag\\
&=\frac{d+1}{2d}\left(S(\mc{E})^2
-u(\mc{E}) + \frac{\lVert \bc{E}_{\rm n} \rVert^2}{d+1} - \frac{\lVert 
	\bc{E}_{\rm sdl} \rVert^2}{d-1}\right)\label{eq:sym_anti}
\end{align}
respectively, where we have used Proposition~\ref{prop:norms}. Since both these 
expressions are probabilities we have $0\le p_{as} \le 1$ and $0\le p_{sa} \le 1$.
Taking appropriate linear combinations of these two inequalities will cancel the 
dependence on either $\lVert \bc{E}_{\rm n}\rVert^2$ or 
$\lVert \bc{E}_{\rm sdl}\rVert^2$, isolating the other variable. Simplifying the resulting 
expressions gives the bound \autoref{eq:bound_norms} for both quantities 
individually, hence the maximum holds as well.

Furthermore, if the noise is trace-preserving, then $\lVert \bc{E}_{\rm sdl} 
\rVert^2=0$ and $S(\mc{E})=1$, so $p_{as} \ge 0$ gives 
$\lVert \bc{E}_{\rm n}\rVert^2 \leq (d-1)[1-u(\mc{E})]$ 
for trace-preserving noise. 
\end{proof}

We now prove that $u(\mc{E})=1$ if and only if $\mc{E}$ is unitary and that 
$u(\mc{E})$ is invariant under composition with unitaries.

\begin{prop}\label{prop:unitarity}
	For any channel $\mc{E}$, $u(\mc{E})\leq 1$ with equality if and only 
	if $\mc{E}$ is unitary. Furthermore, the unitarity satisfies 
	$u(\mc{V}\circ\mc{E}\circ\mc{U}) = u(\mc{E})$ for any unitaries $U,V\in 
	U(d)$.
\end{prop}

\begin{proof}
The unitary invariance $u(\mc{V}\circ\mc{E}\circ\mc{U}) = u(\mc{E})$ follows 
immediately from the invariance of the trace under cyclic permutations.

Since the norms of vectors are always nonnegative, $u(\mc{E}) = 1$ only if 
$\mc{E}$ is trace-preserving and unital by \autoref{eq:bound_norms}, in which 
case the adjoint channel $\mc{E}\ct$ is also a channel~\cite{Watrous2011} and 
so the eigenvalues of $\bc{E}\ct\bc{E}$ (i.e., the singular values of $\bc{E}$) 
are all bounded by one~\cite{Evans1978}. Therefore $u(\mc{E})=1$ only if 
$\mc{E}$ is unital and all the eigenvalues of $\bc{E}$ have unit modulus and 
consequently if $\av{\det\bc{E}}=1$. However, the only channels with 
$\lvert\det\bc{E}\rvert = 1$ are unitary channels~\cite{Wolf2008}. Since 
$u(\mc{E})$ is unitarily invariant and $u(\mc{I})=1$, $u(\mc{E})=1$ if and only 
if $\mc{E}$ is unitary, as claimed.
\end{proof}

We now show that the unitarity can be used with the average gate infidelity to 
quantify the intermediate regime between incoherent and unitary errors. It is 
useful to define a notion of average gate infidelity that has been optimized to 
remove unitary noise. First recall the definition of average gate infidelity, 
\begin{align}
	r(\mc{E}) = 1- \int \mathrm{d}\psi \tr[\psi \mathcal{E}(\psi)]\,.
\end{align}
Then for any CPTP channel $\mc{E}$, define 
\begin{align}
	R(\mc{E}) = \min_{U,V\in U(d)} r(\mc{V}\circ\mc{E}\circ\mc{U})\,.
\end{align}
This quantity can be thought of as the best average gate infidelity 
that is achievable with perfect unitary control. For example, 
if $\mc{E}$ is a unitary channel, then $R(\mc{E})=0$. 

\begin{prop}\label{prop:bound}
	For any CPTP channel $\mc{E}$ with average gate infidelity $r=r(\mc{E})$ to the 
	identity and $R=R(\mc{E})$ as above, then the following inequalities hold
	\begin{align}\label{eq:unitarityineq}
		u(\mc{E})\geq [1 - d R/(d-1)]^2 \geq [1 - dr/(d-1)]^2 \,.
	\end{align}
	The chain of inequalities is saturated if and only if $\mc{E}$ has a 
	unital block $\bc{E}_{\rm u}$ that is a diagonal scalar matrix.
\end{prop}

\begin{proof}
Any channel with infidelity $r$ to the identity can be written as $\bc{E} = 
\bc{I} - r\Delta$ where the diagonal entries of $\Delta$ are nonnegative and 
$\tr\Delta = d(d+1)$. We then have
\begin{align*}
\lVert \bc{E}_{\rm u} \rVert_2^2 \geq \sum_{k=2}^{d^2} (1-\Delta_{kk}r)^2 = 
d^2-1 -2r\tr \Delta + \sum_{k=2}^{d^2}\Delta_{kk}^2 r^2,
\end{align*}
with equality if and only if $\bc{E}_{\rm u}$ is diagonal. The term $\sum_k 
\Delta_{kk}^2$ is uniquely minimized for nonnegative $\Delta_{kk}$ subject to 
the constraint $\tr\Delta = d(d+1)$ by setting $\Delta_{kk} = d/(d-1)$ (that 
is, by setting all the diagonal entries to be equal). This proves the weaker 
inequality bounding $u(\mc{E})$ in terms of $r$. To get the stronger inequality 
in terms of $R(\mc{E})$, we use the unitary invariance proven in 
Proposition~\ref{prop:unitarity} and optimize the inequality over all unitary channels. 
\end{proof}

We note that the first inequality in \autoref{eq:unitarityineq} is saturated 
at 1 when the noise channel is unitary, and the chain of inequalities is 
saturated for depolarizing noise, or depolarizing noise composed with amplitude damping.

An immediate corollary of Proposition~\ref{prop:bound} is that the unitarity can 
be used to put a lower bound on the best possible average infidelity in 
the presence of perfect unitary control. Rearranging 
\autoref{eq:unitarityineq}, we find a lower bound
\begin{align}
	\frac{d-1}{d} \bigl(1-\sqrt{u(\mc{E})}\bigr) \le R(\mc{E}) \le r(\mc{E}) \,.
\end{align}

The unitarity is also closely related to the purity of the Jami\l{l}kowski state 
associated to the noise channel.

\begin{prop}\label{prop:Jpure}
The unitarity is related to the purity of the Jamio{\l}kowski state by
\begin{align}\label{eq:Jamiolkowski}
	d^2\tr\left[J(\mc{E})\ct J(\mc{E})\right] &= S(\mc{E})^2 + \|\bs{\E}_{\rm 
	sdl}\|^2 + \|\bs{\E}_{\rm n}\|^2 \notag\\
&\quad + (d^2-1) u(\E)
\end{align}
where $J(\mc{E}) = (\mc{E}\otimes\mc{I})[\Phi]$ and $\Phi = 
\tfrac{1}{d}\sum_{j,k}\ket{jj}\bra{kk}$.
\end{prop}

\begin{proof}
We begin with an alternate representation of the maximally entangled state $\Phi = 
d^{-1}\sum_k A_k \otimes A_k$. By cycling the adjoint channel in the trace, 
the purity of $J(\mc{E})$ becomes
\begin{align}
	\tr \bigl[J(\mc{E})\ct J(\mc{E})\bigr] 
	&=\tr\bigl[\Phi\, (\mc{E}\ct\mc{E}\otimes\mc{I})[\Phi]\bigr]\notag\\
	&=(\Phi|\bc{E}\ct\bc{E}\otimes\bc{I}|\Phi) \notag\\
	&= \frac{1}{d^2}\sum_{j,k} (A_j\otimes A_j|\bc{E}\ct\bc{E}\otimes \bc{I}|A_k\otimes A_k) 
	\notag\\
	&= \frac{1}{d^2}\sum_{j,k} (A_j|\bc{E}\ct\bc{E}|A_k) (A_j|A_k) \,.
\end{align}
Since the $A_k$ are a trace orthonormal basis, the last line simplifies to
\begin{align}
\tr \bigl[J(\mc{E})\ct J(\mc{E})\bigr] = \tfrac{1}{d^2}\tr \bc{E}\ct\bc{E}.
\end{align}
Comparing this expression to the decomposition in \autoref{eq:block_terms} and using 
\autoref{prop:unitarity_relation} completes the proof. 
\end{proof}

Finally, we give a simple example that shows that the unitarity is \emph{not} 
a monotone, in the sense that it can oscillate under composition of channels.
Consider the two (nearly) dual qubit channels,
\[\mc{E}_0(\rho) = \tr(\rho)\ket{0}\bra{0}
 \quad \text{and}\quad 
 \frac{1}{2} \mc{E}_0^\dagger(\rho) = \bra{0}\rho\ket{0}\frac{\unit}{2}\,.
 \]
Then the unitarity of both $\mc{E}_0$ and $\frac12\mc{E}_0^\dagger$ is zero, 
while the unitarity of the composed channel 
$\frac12\bc{E}^{\vphantom{\dagger}}_0 \bc{E}_0^\dagger$ is 1/12.

We note that for some restricted classes of channels the unitarity is indeed a 
monotone. For example, a trivial application of von Neumann's trace inequality 
shows that if the singular values of the unital block are all 
less than or equal to 1 (which holds for all qubit channels and all unital 
channels), then it is a monotone for trace-preserving channels. 

\section{Conclusion}

In this paper, we have shown that the coherence of a noisy process can be 
quantified by the unitarity, which corresponds to the change in the purity 
(with the identity components subtracted off) averaged over pure states. We 
have presented a protocol for efficiently estimating the unitarity of the 
average noise in the implementation of a unitary 2-design.

We have also proven that the unitarity is 1 if and only if the noise source is 
unitary and provided a tight lower bound for the unitarity in terms of the 
infidelity (which can be estimated using randomized 
benchmarking~\cite{Magesan2011}). This allows the intermediate regime between 
fully incoherent and unitary errors to be quantified, potentially allowing for 
improved bounds on the worst-case error. We have also shown that the unitarity 
provides a lower bound on the best achievable gate infidelity assuming perfect 
unitary control.

Our present results also have direct implications for the loss protocol when 
applied to a unitary 2-design, since the variance over 
random sequences of fixed length for the protocol in Ref.~\cite{Wallman2014a} is
\begin{align}
\mathds{V}_{\bf{j}}(Q_{\bf{j}}) = \mathds{E}_{\bf{j}}(Q_{\bf{j}}^2) - 
[\mathds{E}_{\bf{j}}(Q_{\bf{j}})]^2,
\end{align}
which decays faster with $m$ for fixed $S(\mc{E})$ if the unitarity is 
smaller (and hence the two decay rates in the fit curve for determining the unitarity, 
$\lambda_{\pm}$, are smaller). A lower variance over 
sequences allows a more precise estimation of the average incoherent survival 
probability for a fixed number of experiments. Similar implications may also 
hold for standard randomized benchmarking since $u(\mc{E})$ can easily be seen 
to be one of the eigenvalues of the averaged operator in 
Ref.~\cite{Wallman2014} that determines the variance and is precisely the 
eigenvalue that determines the asymptotic variance. However, in order to 
establish a concrete bound, it would have to be shown that $u(\mc{E})$ is in 
fact the largest eigenvalue.

There are four important open problems raised by this work. First, while the 
unitarity is a monotone for unital noise, it is not a monotone for 
trace-decreasing noise. We leave open the problem of finding necessary and 
sufficient conditions for when $u(\mc{E})$ is a monotone, or finding other 
quantities that are monotonic in general.

Second, our protocol characterizes the unitarity of the average noise, but does 
not characterize the unitarity of the errors in the individual gate. While
a variant of interleaved randomized benchmarking~\cite{Magesan2012} should hold
for the current protocol, obtaining reasonable bounds on the unitarity of the 
individual error is an open problem.

Third, the signal for our protocol is substantially improved by 
the purity measurement, but the method of performing the purity measurement via 
measuring Pauli operators is not scalable beyond a handful of qubits because of the 
exponential size of the Pauli group on $n$ qubits. Moreover, measuring any 
single Pauli operator will in general give a small signal as the number of qubits grows, 
since we do not perform an inversion step. Directly using the SWAP operation on two 
copies of the system running in parallel is a mathematical solution, but the extra 
resources required to implement this might be prohibitive and an analysis of the role 
of crosstalk and correlations would be required to justify this idea. 
Thus, identifying efficient measurements that give a good signal on multi-qubit 
systems remains an open problem.

Finally, a pressing open problem identified in this paper is to obtain an 
improved bound on the worst-case error in terms of both the infidelity and the 
unitarity. Such a bound would substantially reduce the effort 
required to certify that an experimental implementation is near (or below) 
the threshold for fault-tolerant quantum computation.

\acknowledgments 

We thank Jay Gambetta for pointing out that estimating purity would reduce variance.
This work was supported by the ARC via EQuS project number CE11001013, 
by IARPA via the MQCO program, and by 
the US Army Research Office grant numbers W911NF-14-1-0098 and W911NF-14-1-0103. 
STF also acknowledges support from an ARC Future Fellowship FT130101744.

\bibliography{library}

\end{document}